\documentclass[conference]{IEEEtran}
\IEEEoverridecommandlockouts
\usepackage{graphicx}
\usepackage{array}
\usepackage{amsmath}
\usepackage{float}
\usepackage{multirow}
\usepackage{array}
\usepackage{cancel}
\usepackage{tabu}
\usepackage{amsthm}
\usepackage{epsfig}
\usepackage{epstopdf}
\usepackage{epsf}
\usepackage{color}
\usepackage[english]{babel}

\theoremstyle{definition}
\newtheorem{theorem}{Theorem}

\usepackage{amssymb}
\usepackage{cite}
\usepackage{amsmath,amssymb,amsfonts}
\usepackage{graphicx}
\usepackage{textcomp}
\usepackage{xcolor}
\usepackage{algorithm}
\usepackage[noend]{algpseudocode}
\usepackage[T1]{fontenc}
\usepackage{graphicx,lipsum,afterpage,subcaption}
\usepackage{enumitem}
\usepackage{lipsum}
\usepackage{mathtools}
\usepackage{cuted}

\DeclareMathOperator*{\argmax}{argmax}
\def\BibTeX{{\rm B\kern-.05em{\sc i\kern-.025em b}\kern-.08em
    T\kern-.1667em\lower.7ex\hbox{E}\kern-.125emX}}
\begin{document}

\title{Single-bit Quantization's Capacity of Binary-input Continuous-output Channels \\
}

	\author{\IEEEauthorblockN{Thuan Nguyen}
		\IEEEauthorblockA{School of Electrical and\\Computer Engineering\\
			Oregon State University\\
			Corvallis, OR, 97331\\
			Email: nguyeth9@oregonstate.edu}

		\and
		\IEEEauthorblockN{Thinh Nguyen}
		\IEEEauthorblockA{School of Electrical and\\Computer Engineering\\
			Oregon State University\\
			Corvallis, 97331 \\
			Email: thinhq@eecs.oregonstate.edu}
	}

\maketitle

\begin{abstract}
We consider a channel with discrete binary input X that is
corrupted by a given continuous noise to produce a continuous-valued
output Y.  A quantizer is then used to quantize the
continuous-valued output Y to the final binary output Z. The goal is
to design an optimal quantizer $Q^*$ and also find the optimal input distribution $p^*_X$ that maximizes the mutual information $I(X; Z)$ between the binary input and the binary quantized output. A linear time complexity searching procedure is proposed. Based on the properties of the optimal quantizer and the optimal input distribution, we reduced the searching range that results in a faster implementation algorithm. Both theoretical and numerical results are provided to illustrate our method. 
\end{abstract}


\section{Introduction and Related Work}
\label{sec: related work}
	A communication system can be modeled by an abstract channel with a set of inputs at the transmitter and a set of corresponding outputs at the receiver. Often times the transmitted symbols (inputs) are different from the receiving symbols (outputs), i.e., errors occur due to many factors such as the physics of signal propagation through a medium or thermal noise. Thus, the goal of a communication system is to transmit the information reliably at the fastest rate.  The fastest achievable rate with vanishing  error for a given channel is defined by its channel capacity which is the maximum mutual information between input and output random variables.  For an arbitrary discrete memoryless channel (DMC) that is specified by a given channel matrix, the mutual information is a concave function of the input probability mass function \cite{cover2012elements}.  Thus, many efficient algorithms exist to find the channel capacity of DMC  \cite{blahut1972computation}.  Moreover, under some special conditions of the channel matrix, the closed-form expressions of channel capacity can be constructed.  \cite{nguyen2018closed}. It is worth noting that due to the simplicity of binary channels, the closed-form expressions of channel capacity of a binary channel always can be found as a function of the diagonal entries of the channel matrix \cite{moskowitz2010approximations}, \cite{moskowitz2009approximation},  \cite{silverman1955binary}. 
	
On the other hand, in many real-world scenarios, the input distribution is given, however, one has to design the channel matrix under the consideration of many factors such as power consumption, encoding/decoding speeds, and so on.  As a result, the mutual information is no longer a concave function of the input distribution alone but is a possibly non-concave/convex function in both input distribution and the parameters of the channel matrix.  Many advanced quantization algorithms have also been proposed over the past decade \cite{kurkoski2014quantization}, \cite{winkelbauer2013channel}, \cite{iwata2014quantizer}, \cite{he2019dynamic}, \cite{sakai2014suboptimal},  \cite{koch2013low} to find the optimal quantizer under the assumption of given the input distribution.   These algorithms play an important role in designing Polar code and LDPC code decoders \cite{romero2015decoding}, \cite{tal2011construct}.

Recently, there are many works on designing quantizers together with finding the optimal input distribution  such that the mutual information over both quantization parameters and the input probability mass function is maximized.  Although the mutual information is a concave function in the input pmf, it is not a convex/concave function in the quantization parameters i.e., thresholds. Therefore, many famous convex optimization techniques and algorithms for finding the global optimal solution are not applicable. To our best knowledge, this problem remains to be a hard and not well-studied  \cite{mathar2013threshold},  \cite{nguyen2018capacities},  \cite{alirezaei2015optimum}, \cite{singh2009limits}, \cite{kurkoski2012finding}.  In \cite{singh2009limits}, Singh et al. provided an algorithm for multilevel quantization,
which gave near-optimal results. In \cite{nguyen2018capacities}, the author proposed a heuristic near-optimal quantization algorithm.  However, this algorithm only works well when the SNR ratio of the channel is high.  For 1-bit quantization of general additive channels, Alirezaei and Mathar showed that capacity could be achieved by using an input distribution with only two support points \cite{mathar2013threshold}. In \cite{kurkoski2012finding}, the author gave a near-optimal algorithm to find the optimal of mutual information over both input distribution and quantizer variables for binary input and an arbitrary number of the quantized output, however, this algorithm may declare a failure outcome. 

In this paper, we provide a linear time complexity searching procedure to find the global optimal of mutual information between input and quantized output over both input distribution and quantizer variables. Based on the properties of the optimal quantizer and the optimal input distribution, the searching range is reduced that finally results in a faster implementation algorithm. Both the theoretical and numerical results are provided to justify our contributions. 


\section{Problem description}
\label{sec:problem description}
\begin{figure}
	\centering
	\includegraphics[width=0.9\linewidth]{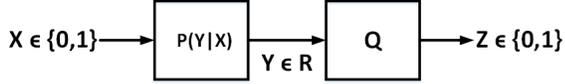}
	\caption{Channel model: a binary input $X=\{0,1\}$ is corrupted by continuous noise to result in continuous-valued $Y=\mathbb{R}$ at the receiver. The receiver attempts to recover $X$ by quantizing $Y$ into binary signal $Z=\{0,1\}$.}
	\label{fig:setup}
\end{figure}

We consider the channel shown in Fig. \ref{fig:setup} where the binary signals $ x \in X=\{0,1\}$ having $p_X=\{p_{x=0},p_{x=1}\}$ are transmitted and corrupted by a continuous noise source to produce a continuous-valued output $y \in \mathbb{R}$ at the receiver. Specifically, $y$ is specified by the a channel conditional density $p(y|x)$ where $p(y|x)$ models the distortion caused by noise. The receiver recovers the original binary signal $x$ using a quantizer $Q$ that  quantizes the received continuous-valued signal $y$ to $z \in Z=\{0,1\}$. 
Since $y\in \mathbb{R}$, the quantization parameters can be specified by a thresholding vector 

$$\textbf{h}= (h_1, h_2, \dots,h_n) \in \mathbb{R}^n,$$
with $h_1 < h_2 < \dots < h_{n-1} < h_n$, where $n$ is assumed a finite number.  Theoretically, it is possible to construct the conditional densities $p(y|x)$ such that the optimal quantizer might consist an infinite number of thresholds.  However, for a practical implementation, especially when the quantizer is implemented using a lookup table, then a finite number of thresholds must be used.  To that end, in this paper, we assume that the quantizer using an finite number of thresholds. Now, $\textbf{h}$ induces $n+1$ disjoint partitions: $$H_1 = (-\infty, h_1), H_2 = [h_1, h_2), \dots, H_{n+1} = [h_n,\infty).$$ 
Let $\mathbb{H} = \bigcup_{i \in odd} H_i$ and  $\bar{\mathbb{H}} = \bigcup_{i \in even} H_i$, then $\mathbb{H} \cap \bar{\mathbb{H}} =\emptyset$ and $\mathbb{H} \cup \bar{\mathbb{H}} =\mathbb{R}$. Thus, $\textbf{h}^*$ divides $\mathbb{R}$ into $n+1$ contiguous disjoint segments, each maps to either 0 or 1 alternatively
Without the loss of generality, we suppose that the receiver uses a quantizer $Q:Y \rightarrow Z$ to quantize $Y$ to $Z$ as:
\begin{equation}
 \label{eq:decoding}
Z = \begin{cases} 
	0  & \text{if  } Y \in \mathbb{H}, \\
	1  & \text{if  }  Y \in \bar{\mathbb{H}}.
\end{cases}
\end{equation}
Our goal is to design an optimal quantizer $Q^*$, specifically $\textbf{h}^*$ and also find the optimal input distribution $p_X^*$ that maximizes the mutual information $I(X;Z)$ between the input $X$ and the quantized output $Z$:
\begin{equation}
\label{eq:maximization}
h^*,p_X^*= \argmax_{h^*,p_X^*} I(X;Z).
\end{equation}
We note that the values of thresholds $h_i$'s, the number of thresholds $n$ and input distribution $p_X$ are the optimization variables. The maximization in (\ref{eq:maximization}) only assumes that the channel conditional density $p(y|x)$ are given. 
\section{Optimality conditions}
\label{sec: optimal conditions}
For convenience, we use the following notations:
\begin{enumerate}
	\item $p = (p_0,  p_1)$  denotes the probability mass function for the input $X$, with $p_0 =  P(X=0)$ and $p_1 = P(X=1)$.
	\item $q = (q_0, q_1)$ denotes probability mass function for the output $Z$, with  $q_0 =  P(Z=0)$ and $q_1 = P(Z=1)$.
	\item  $\phi_0(y)=p(y|x=0)$ and $\phi_1(y)=p(y|x=1)$ denote conditional density functions of the received signal $Y$ given the input signal $X = 0$ and $X=1$, respectively.
\end{enumerate}

The 2$\times2$ channel matrix $A$ associated with a discrete memoryless channel (DMC) with input $X$ and output $Z$ is:
\[
A=
\begin{bmatrix}
A_{11} & 1-A_{11}\\
1-A_{22} & A_{22}
\end{bmatrix},
\]

where
\begin{equation}
\label{eq: definition a11}
A_{11} =\int_{y \in \mathbb{H}}^{}\phi_0(y)dy,
\end{equation}
\begin{equation}
\label{eq: definition a22}
A_{22} =\int_{y \in \bar{\mathbb{H}}}^{}\phi_1(y)dy.
\end{equation}

\subsection{Optimal quantizer structure for a given input distribution}
\label{sec:structure}
 Our first contribution is to show that for a given input distribution the optimal binary quantizer with multiple thresholds, specified by a thresholding vector $\textbf{h}^*=(h_1^*,h_2^*,\dots,h_n^*)$ with $h_i^* < h_{i+1}^*$, must satisfy the conditions stated in the Theorem \ref{theorem: 1} below.  
  
\begin{theorem}
\label{theorem: 1}
Let $\textbf{h}^*=(h_1^*,\dots,h_n^*)$ be a thresholding vector of an optimal quantizer $Q^*$, then:
\begin{equation}
\label{eq: relation threshold}
\dfrac{\phi_0(h_i^*)}{\phi_1(h_i^*)} = \dfrac{\phi_0(h_j^*)}{\phi_1(h_j^*)} = r^*,
\end{equation}
for $\forall$ $i, j \in \{1,2,\dots,n\}$ and some optimal constant $r^* > 0$. 
\end{theorem}
\begin{proof}
We note that using the optimal thresholding vector $\textbf{h}^*$, the quantization mapping follows 
(\ref{eq:decoding}).  $\textbf{h}^*$ divides $\mathbb{R}$ into $n+1$ contiguous disjoint segments, each maps to either 0 or 1 alternatively.
The discrete memoryless channel in Fig. \ref{fig:setup} has the channel matrix 
\[
A^*=
  \begin{bmatrix}
A_{11} & A_{12}\\
A_{21} & A_{22}
  \end{bmatrix},
\]
and the mutual information can be written as a function of $h$ as: 
\begin{equation}
I(h)=H(Z)-H(Z|X)=H(q_0)-[p_0H(A_{11})+p_1H(A_{22})],
\end{equation}
where for any $w \in [0,1]$, $H(w)=-[w\log(w)+(1-w)\log(1-w)]$ and $q_0=P(Z=0)=p_0A_{11}+p_1A_{21}$.

This is an optimization problem that maximizes $I(h)$.  The theory of optimization requires that an optimal point must satisfy the KKT conditions \cite{boyd2004convex}.  In particular, define the Lagrangian function as:
\begin{equation}
\label{eq:lagrangian}
L(h,\lambda) = I(h) + \sum_{i=1}^{n-1}{\lambda_i (h_i - h_{i+1})},
\end{equation}
then the KKT conditions \cite{boyd2004convex} states that, an optimal point $h^*$ must satisfy:

\begin{equation}
\label{eq:kkt1}
\begin{cases}
\frac{d{L(h, \lambda)}}{dh}|_{h = h^*, \lambda = \lambda^*} = 0, \\
\lambda_i^*(h_i - h_{i+1}) = 0, i = 1, 2, \dots, n-1, \\
\lambda_i^* \ge 0, i = 1, 2, \dots, n-1.
\end{cases}
\end{equation}
Since the structure of the quantizer requires that $h_i < h_{i+1}$, the second and the third conditions in (\ref{eq:kkt1}) together imply that $\lambda_i^* = 0, i = 1, 2, \dots, n-1$. Consequently, from (\ref{eq:lagrangian}) and the first condition in (\ref{eq:kkt1}), we have:

$$\frac{d{L(h, \lambda)}}{dh}|_{h = h^*, \lambda = \lambda^*} = \frac{d{I(h)}}{dh}|_{h = h^*} = 0.$$

By setting the partial derivatives of $I(h)$ with respect to each $h_i$ to zero, we have
\begin{strip}
\vspace{-0.1 in}
\hrule
\begin{eqnarray}
 \frac{\partial I(h)}{\partial h_i} &=&(\log\dfrac{1-q_0}{q_0} )\frac{\partial q_0}{\partial h_i} -p_0(\log \dfrac{1-A_{11}}{A_{11}}) \frac{\partial A_{11}}{\partial h_i}-p_1(\log\dfrac{1-A_{22}}{A_{22}})\frac{\partial A_{22}}{\partial h_i} \nonumber\\
&=&(\log\dfrac{1-q_0}{q_0})(p_0 \frac{\partial A_{11}}{\partial h_i} - p_1 \frac{\partial A_{22}}{\partial h_i} )-p_0(\log\dfrac{1-A_{11}}{A_{11}}) \frac{\partial A_{11}}{\partial h_i} -p_1(\log\dfrac{1-A_{22}}{A_{22}})\frac{\partial A_{22}}{\partial h_i}\label{eq: 2}\\
&=&p_0 \frac{\partial A_{11}}{\partial h_i} (\log \dfrac{1-q_0}{q_0}-\log \dfrac{1-A_{11}}{A_{11}})-p_1 \frac{\partial A_{22}}{\partial h_i} (\log \dfrac{1-q_0}{q_0}+\log \dfrac{1-A_{22}}{A_{22}}) = 0,
\label{eq:derivative}
\end{eqnarray}
\hrule
\vspace{-0.1 in}
\end{strip}
with (\ref{eq: 2}) due to $q_0=p_0A_{11}+p_1A_{21}=p_0A_{11}+p_1(1-A_{22})$.

Since  $\frac{\partial A_{11}}{\partial h_i}=\phi_0(h_i)$ and $\frac{\partial A_{22}}{\partial h_i}=-\phi_1(h_i)$, from (\ref{eq:derivative}), we have:
\begin{equation}
\label{eq: 3}
\dfrac{\phi_0(h^*_i)}{\phi_1(h^*_i)}=-\dfrac{p_1}{p_0}\dfrac{\log\dfrac{1-q_0}{q_0}+\log\dfrac{1-A_{22}}{A_{22}}}{\log\dfrac{1-q_0}{q_0}-\log\dfrac{1-A_{11}}{A_{11}}}= r^*.
\end{equation}
Since (\ref{eq: 3}) holds for $\forall$ $i$, the RHS of (\ref{eq: 3}) equals to some constant $r^* > 0 $ for a quantizer $Q^*$,  Theorem \ref{theorem: 1} follows.
\end{proof}
Suppose the optimal value $r^*$ is given and the equation $r(y) = r^*$ has $m$ solutions: $y_1 < y_2 < \dots <  y_m$. Then,  Theorem \ref{theorem: 1} says that the optimal quantizer must either have its thresholding vector be $(y_1, y_2, \dots, y_m)$ or one of its ordered subsets, e.g., $(h^*_1, h^*_2) = (y_1, y_3)$, or both.   In Theorem \ref{theorem: 1-a} below, we will show that the quantizer whose thresholding vector contains all the solutions of $\dfrac{\phi_0(y)}{\phi_1(y)} = r^*$, will be at least as good as any quantizer whose thresholding vector is a ordered subset of the set of all solutions.  
\begin{theorem}
\label{theorem: 1-a}
Let $y^*_1 < y^*_2< \dots < y^*_n$ be the solutions of $r(y) =\dfrac{\phi_0(y)}{\phi_1(y)}= r^*$ for the optimal constant $r^* > 0$. Let $Q^n_{r^*}$ be the quantizer whose thresholding vector contains all the solutions, i.e., $h^*_i = y^*_i, i = 1, 2, \dots, n$, then for $k < n$, $Q^n_{r^*}$ is at least as good as any quantizer $Q^k_{r^*}$ whose thresholding vector is an ordered subset of the set of $(h^*_1, h^*_2, \dots, h^*_n)$.
\end{theorem}  
\begin{proof}
Due to the limited space, we do not present the detailed proof of Theorem \ref{theorem: 1-a}. However, we refer the reader to Theorem 1 in \cite{kurkoski2017single}. Indeed, Theorem 1 in \cite{kurkoski2017single} showed that the optimal quantizer is equivalent to hyper-plane cuts in the space of posterior conditional distribution $p_{y|x}$ and it guarantees that at least one of the globally optimal quantizers has this structure. Due to the channel is binary, a hyper-plane in the posterior distribution $p_{y|x}$ is a point. That said, the optimal threshold vector $h_i$ should be the solutions of 
\begin{equation*}
p_{y|x_1}=\dfrac{\phi_0(y)}{\phi_0(y)+\phi_1(y)}=a^*, 0 \leq a^* \leq 1,
\end{equation*}
 $\forall$ $i=1,2,\dots,n$ that, in turn, is equivalent to $\dfrac{\phi_0(y)}{\phi_1(y)}=r^*$ where $r^*=\dfrac{a}{1-a}$. 
\end{proof}
Theorem \ref{theorem: 1-a} is important in the sense that it provides a concrete approach to find the globally optimal quantizer $Q^*$ by exhausted searching the optimal $r^*$ and using all the solutions of $\dfrac{\phi_0(y)}{\phi_1(y)}=r^*$ to construct the optimal threshold vector $\textbf{h}$. We also note that Theorem \ref{theorem: 1-a} can stand alone without using the proof of Theorem \ref{theorem: 1}, however, \textit{Theorem \ref{theorem: 1} is useful in the sense that it provides an important connection between the optimal thresholds that produces the optimal channel matrix and the optimal input distribution.} For example, the relationship in (\ref{eq: 3}) will be used in the following Theorem \ref{theorem: 2}. 

\begin{theorem}
\label{theorem: 2}
Consider a binary channel with a given input distribution, corresponding to an optimal quantizer $Q^*$, the optimal channel matrix $A^*$ having diagonal entries $A_{11}$ and $A_{22}$ such that  $A_{11} \geq p_0$ and $A_{22} \geq p_1$.  
\end{theorem}
\begin{proof}
Please see the appendix. 
\end{proof}

\vspace{-0.1 in}
\subsection{Optimal input distribution for a given channel matrix}
 For a binary channel having a given channel matrix $A$, the optimal input distribution can be determined in closed-form \cite{moskowitz2010approximations}, \cite{moskowitz2009approximation}. Finally, the maximum of mutual information at the optimal distribution can be written as a function of channel matrix entries \cite{moskowitz2010approximations}, \cite{moskowitz2009approximation},  \cite{silverman1955binary}, \cite{nguyen2018closed}. This result is summarized  in the following Theorem.
\begin{theorem}
\label{theorem: 3}
For a given quantizer $Q$ which corresponds to a given channel matrix $A$, the maximum of mutual information $I(X;Z)$ can be written by the following closed-form
\begin{small}
\begin{eqnarray}
\label{eq: closed form with fix channel matrix}
I(X;Z)_{p_X^*} &=&\log_{2}[2^{-\dfrac{A_{22}H(A_{11})+(A_{11}-1)H(A_{22}))}{A_{11}+A_{22}-1}} \nonumber\\
&  + & 2^{-\dfrac{(A_{22}-1)H(A_{11})+A_{11}H(A_{22}))}{A_{11}+A_{22}-1}}  ],
\end{eqnarray}
\end{small}
where $H(w)=- [w\log(w) +(1-w)\log(1-w)]$. 
\end{theorem}
\begin{proof}
Please see the detailed proof in \cite{moskowitz2010approximations}, \cite{moskowitz2009approximation},  \cite{silverman1955binary}, \cite{nguyen2018closed}. 
\end{proof}
\begin{theorem}
\label{theorem: 4}
For a given quantizer $Q$ which corresponds to a given channel matrix $A$, the optimal input distribution $p_0^*$ and $p_1^*$ are bounded by:
\begin{equation}
0.3679=\dfrac{1}{e}< q_0^*,q_1^*<1-\dfrac{1}{e}=0.6321.
\end{equation}
\end{theorem}
\begin{proof}
Please see Theorem 1 in \cite{majani1991two}. 
\end{proof}

\section{Finding channel capacity over both input distribution and threshold vector variables}
\label{sec: solution}
Theorem \ref{theorem: 1} and Theorem \ref{theorem: 1-a} state that an optimal quantizer can be found by exhaustive searching the optimal value $r^*$ and use all the solutions of  $\dfrac{\phi_0(y)}{\phi_1(y)} = r^*$ to construct the optimal thresholding vector $\textbf{h}^*$. The mutual information $I(X;Z)$, therefore, becomes a function of variable $r$. Now, for a given $r$, define $\mathbb{H}_r = \{y: \dfrac{\phi_0(y)}{\phi_1(y)} > r \}$, then
$$ \mathbb{H}_r =  \{ (-\infty, h_1) \ \cup [h_2, h_3) \cup \dots \cup [h_n, +\infty) \}.$$

Similarly,  let $\bar{\mathbb{H}}_r = \{y:  \dfrac{\phi_0(y)}{\phi_1(y)} \leq  r \}$, then
$$\bar{\mathbb{H}}_r = \mathbb{R} \setminus \mathbb{H}_r = \{ [h_1, h_2 ) \cup [h_3, h_4) \cup \dots \cup [h_{n-1}, h_n) \}.$$

The sets $\mathbb{H}_r$ and $\bar{\mathbb{H}}_r$ together specify a binary quantizer that maps $y$ to $z \in \{0,1\}$, depending on whether $y$ belongs to $\mathbb{H}_r$ or $\bar{\mathbb{H}}_r$. Without the loss of generality, suppose we use the following quantizer:
\begin{equation}
z = 
\begin{cases}
0 & y \in \mathbb{H}_r, \\
1 & y \in \bar{\mathbb{H}}_r,
\end{cases}
\end{equation}
then the channel matrix of the overall DMC is:
$$\begin{array}{cc}
A = \begin{bmatrix} f(r) & 1-f(r) \\
1-g(r)  &  g(r)
\end{bmatrix},
\end{array}$$
where $f(r) \stackrel{\triangle}{=} p(z=0|x=0)=A_{11}$ and $g(r) \stackrel{\triangle}{=} p(z = 1|x=1) =A_{22}$.  $f(r)$ and $g(r)$ can be written in terms of $\phi_0(y)$ and $\phi_1(y)$ as:

\begin{footnotesize}
\begin{equation}
\label{eq: construct fr}
f(r)\!=\!\int_{y \in \mathbb{H}_r}\phi_0(y)dy\!=\!\int_{-\infty}^{h_1}\phi_0(y)dy \!+\! \int_{h_2}^{h_3}\phi_0(y)dy \!+\! \dots \!+\! \int_{h_{n}}^{+\infty}\phi_0(y)dy,
\end{equation}

\begin{equation}
\label{eq: construct gr}
g(r)\!=\!\int_{y \in \bar{\mathbb{H}}_r}\phi_1(y)dy\!=\!\int_{h_1}^{h_2}\phi_1(y)dy \!+\! \int_{h_3}^{h_4}\phi_1(y)dy \!+\! \dots \!+\! \int_{h_{n-1}}^{h_{n}}\phi_1(y)dy.
\end{equation}
\end{footnotesize}

Using Theorem \ref{theorem: 3}, the optimal of mutual information in Eq. (\ref{eq: closed form with fix channel matrix}) is:
\begin{small}
\begin{eqnarray}
\label{eq: closed form with fix channel matrix 2}
I(X;Z)_{(p_X^*,r)} &=&\log_{2}[2^{-\dfrac{g(r)H(f(r))+(f(r)-1)H(g(r)))}{f(r)+g(r)-1}} \nonumber\\
&  + & 2^{-\dfrac{(g(r)-1)H(f(r))+f(r)H(g(r)))}{f(r)+g(r)-1}}  ].
\end{eqnarray}
\end{small}

\textbf{Linear time complexity algorithm:} using (\ref{eq: closed form with fix channel matrix 2}), an exhausted searching over $r$ can be applied to find the optimal of mutual information for both input distribution and threshold quantization. We note that $f(r)$ and $g(r)$ can be computed using (\ref{eq: construct fr}) and (\ref{eq: construct gr}) where $\textbf{h}=\{h_1,h_2,\dots,h_n\}$ are well defined as the solutions of $\dfrac{\phi_0(y)}{\phi_1(y)} = r$.

\textbf{Narrow down the searching area:} 
\begin{theorem}
\label{theorem: 5}
For an arbitrary binary channel, suppose that the optimal of mutual information $I(X;Z)$ over both input distribution and quantizer is achieved at the optimal quantizer $Q^*$ which generates optimal channel matrix $A^*$, then $A_{11}^* > \dfrac{1}{e}$ and $A_{22}^* > \dfrac{1}{e}$.
\end{theorem}
\begin{proof}
Combining Theorem \ref{theorem: 2} and Theorem \ref{theorem: 4}, for an optimal quantizer, we should have $A^*_{11} > \dfrac{1}{e}$ and $A^*_{22} > \dfrac{1}{e}$.
\end{proof}

\begin{theorem}
\label{theorem: 6}
$f(r)$ in (\ref{eq: construct fr}) is a monotonic decreasing function and $g(r)$ in (\ref{eq: construct gr}) is monotonic increasing function with variable $r$. 
\end{theorem}
\begin{proof}
Please see our appendix.
\end{proof}

From Theorem \ref{theorem: 5}, the entries $A_{11}$ and $A_{22}$ should be satisfy $A_{11} > \dfrac{1}{e}$ and $A_{22} > \dfrac{1}{e}$. Thus, we can narrow down the searching range by limiting $f(r)$ and $g(r)$ such that $f(r) >  \dfrac{1}{e}$ and $g(r) > \dfrac{1}{e}$. Due to the monotonic increasing/decreasing of $f(r)$ and $g(r)$, we can find the upper bound and lower bound of $r$ by solving two equations $f(r)=1/e$ and $g(r)=1/e$. Using bisection search, finding the solutions of $f(r)=1/e$ and $g(r)=1/e$ takes the time complexity of $O(\log M)$ where $M=\dfrac{1}{\epsilon}$ and $\epsilon$ is the resolution/accuracy of the solution.

\section{Numerical Results}
\label{sec: simulations}
In this section, we find the optimal of mutual information $I(X;Z)$ for a channel having $\phi_0=N(\mu_0=-1,\sigma_0=6)$ and $\phi_1=N(\mu_1=1,\sigma_1=5)$. Due to $f(r) > 1/e$ and $g(r)> 1/e$, we can limit the searching area of  $r \in [0.8;9.1]$. Next, an exhaustive searching with the resolution $\epsilon=0.01$ over $[0.8;9.1]$ is performed.  Fig. \ref{fig: 6} illustrates the function of $I(X;Z)_{(p_X^*,r)}$ in Eq. (\ref{eq: closed form with fix channel matrix 2}) using variable $r$. From our simulation, the optimal of $I(X;Z)$ for both input variable and threshold variable is $I(X;Z)^*=0.7249$ at $r^*=1.36$. 

   \begin{figure}
  \centering
  \includegraphics[width=3 in]{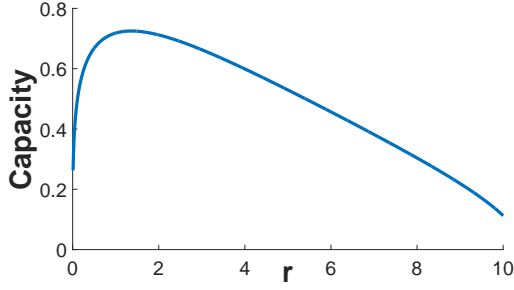}\\
  \caption{$I(X;Z)_{(p_X^*,r)}$ as a function of $r$.}\label{fig: 6}
 \end{figure}

\section{Conclusion}
In this paper, we provide a linear time complexity searching procedure to find the global optimal of mutual information between input and quantized output over both input distribution and quantizer variables. Based on the properties of the optimal quantizer and the optimal input distribution, we reduced the searching range that finally results in a faster implementation. Both theoretical and numerical results are provided to justify our method. 

\appendix
\vspace{-0.1 in}
\subsection{Proof of Theorem 2}
Due to both distribution functions $\phi_0(y)$ and $\phi_1(y)$ are positive, thus, $\dfrac{\phi_0(h_i^*)}{\phi_1(h_i^*)} \geq 0$. From (\ref{eq: 3}),  we have:
\begin{equation}
\label{eq: 4}
-\dfrac{p_1}{p_0}\dfrac{\log\dfrac{1-q_0}{q_0}+\log\dfrac{1-A_{22}}{A_{22}}}{\log\dfrac{1-q_0}{q_0}-\log\dfrac{1-A_{11}}{A_{11}}} \geq 0.
\end{equation}
Using a little bit of algebra, (\ref{eq: 4}) is equivalent to 
\begin{equation}
\label{eq: 5}
(A_{11}-p_0)(A_{22}-p_1) \geq 0.
\end{equation}

Next, we show that $A_{11}+A_{22} \geq 1 =p_0+p_1$. Indeed, $f(r)$ and $g(r)$ represent the quantized bits ``0" and ``1" which correspond to the areas of $\dfrac{\phi_0(y)}{\phi_1(y)} > r$ and $\dfrac{\phi_0(y)}{\phi_1(y)} \leq  r$, respectively. Let $\mathbb{H}_r=\{ y| \dfrac{\phi_0(y)}{\phi_1(y)} > r \}$ and $\bar{\mathbb{H}}_r=\{ y| \dfrac{\phi_0(y)}{\phi_1(y)} \leq r \}$. 

We consider two possible cases: $r \leq 1$ and $r > 1 $.  In both cases, we will show that $A_{11}+A_{22}=f(r) + g(r) \geq 1$.

$\bullet$ If $r \leq 1$ then $\phi_0(y) \leq \phi_1(y)$ for $\forall$ $y \in \bar{\mathbb{H}}_r$.
Therefore,
\begin{eqnarray}
f(r)+g(r)&=& \int_{y \in \mathbb{H}_r} \phi_0(y)dy + \int_{y \in \bar{\mathbb{H}}_r} \phi_1(y)dy\\
& \geq & \int_{y \in \mathbb{H}_r} \phi_0(y)dy + \int_{y \in \bar{\mathbb{H}}_r} \phi_0(y)dy\\
&=&1. \label{eq: 56}
\end{eqnarray}

$\bullet$ If $r > 1$ then $\phi_0(y) > \phi_1(y)$ for $\forall$ $y \in \mathbb{H}_r$.
Therefore,
\begin{eqnarray}
f(r)+g(r)&=& \int_{y \in \mathbb{H}_r} \phi_0(y)dy + \int_{y \in \bar{\mathbb{H}}_r} \phi_1(y)dy\\
& > & \int_{y \in \mathbb{H}_r} \phi_1(y)dy + \int_{y \in \bar{\mathbb{H}}_r} \phi_1(y)dy\\
&=&1. \label{eq: 55}
\end{eqnarray}

Therefore, $A_{11}+A_{22} \geq 1 =p_0+p_1$. Thus,  (\ref{eq: 5}) is equivalent to $A_{11} \geq p_0$ and $A_{22} \geq p_1$.

\subsection{Proof of Theorem 5}
Due to $f(r)$ represents the quantized bit ``0" which is the area of $\phi_0(y)$ where $\dfrac{\phi_0(y)}{\phi_1(y)} > r$.  Therefore, if $r$ is increasing, $f(r)$ is obviously decreasing or $f'(r) 
\leq 0$.  A similar proof can be established for $g(r)$.

\
\bibliographystyle{unsrt}
\bibliography{sample}

\end{document}